\documentclass[pre,12pt, aps, floatfix]{revtex4-1}

\usepackage{amssymb}
\usepackage{amsthm}
\usepackage{graphicx} 

\newtheorem{definition}{Definition}
\newtheorem{theorem}{Theorem}
\newtheorem{lemma}{Lemma}

\begin{document}

\title{Achieving Multiple Goals via Voluntary Efforts and Motivation Asymmetry}
\author{Eckart Bindewald}
\affiliation{Department of Mathematics, \mbox{Frederick Community College}, Frederick, MD 21702} 

\begin{abstract} The achievement of common goals through voluntary efforts of members of a group can be challenged by the high temptation of individual defection. Here, two-person one-goal assurance games are generalized to $N$-person, $M$-goal achievement games in which group members can have different motivations with respect to the achievement of the different goals. The theoretical performance of groups faced with the challenge of multiple simultaneous goals is analyzed mathematically and computationally. For two-goal scenarios one finds that ``polarized'' as well as ``biased'' groups perform well in the presence of defectors. A special case, called individual purpose games ($N$-person, $N$-goal achievements games where there is a one-to-one mapping between actors and goals for which they have a high achievement motivation) is analyzed in more detail in form of the ``importance of being different theorem''. It is shown that in some individual purpose games, groups of size $N$ can successfully accomplish $N$ goals, such that each group member is highly motivated towards the achievement of one unique goal. The game-theoretic results suggest that multiple goals as well as differences in motivations can, in some cases, correspond to highly effective groups. \end{abstract}

\maketitle

Keywords: stag hunt; goal; achievement game; game theory; coordination game; anti-coordination; importance of being different.

\section{Introduction}

Reaching goals is a key ability of intelligent agents. Reaching a goal in a way that needs the contribution of several agents can be modeled as a game: An assurance game is a game-theoretic model, in which members of a group can choose to spend individual efforts or resources for the achievement of a common goal \citep{sen1969game}. The choice of exerting an effort towards a goal has a cost (a negative utility), the achievement of the goal has a benefit (positive utility) for each group member. The original formulation of the assurance game corresponds to two agents, one goal and two choices per member of contributing a high effort or a low effort towards that goal. Other names for this class of games are coordination game, trust dilemma or stag hunt (based on a hypothetical scenario proposed by the philosopher Jean-Jacques Rousseau in which two hunters can choose to hunt a stag corresponding to a large payoff or a hare corresponding to a small payoff; the catch is that successfully hunting the stag needs both hunters to choose that option) \citep{mcadams2008beyond,skyrms2004stag}.

Such situations can be analyzed within the framework of non-cooperative game theory, where all participants (interchangeably referred as players, actors, agents or persons) can choose between different actions and strategies in order to maximize their expected outcome. The key concept is that of a (Nash) equilibrium, where no player can improve unilaterally the outcome by changing the strategy \citep{nash1951non}.

The classic assurance game has three Nash equilibrium points: two pure-strategy equilibria corresponding to mutual cooperation and to mutual defection as well as one mixed-strategy equilibrium point in which both agents choose between cooperation and defection with a probability $1/2$.

Milinski \textit{et al.} studied experimentally iterated assurance games performed by groups with six members, each of which has fine-grained donation options in order to potentially obtain a common-pool reward provided the combined donations of the group are at least as high as a certain threshold \citep{Milinski2008}. Also, it has been noted that the outcome of collective action challenges can depend not only on the rewards but also on the structure of communication networks \citep{chwe2000communication}.

The games typically analyzed by game-theoretic analysis are dealing with scenarios to reach one particular goal \citep{mccain2010game}. Groups are, however, frequently faced with multiple simultaneous challenges: families have the challenge of raising children and earning money; societies have the challenge of helping those in need, while simultaneously protecting its members from threats. This can be viewed as challenges similar to multi-attribute negotiations, where participants have different motivations with respect to different objectives \citep{lai2004literature}. Multi-attribute game theory has been applied to auctions \citep{bichler2000experimental}, border security patrolling \citep{aguirre2011towards} and supply chain network negotiations \citep{yu2013coalition}.

Because the different choices for the differently motivated participant quickly leads to a ``combinatorial explosion'' of possibilities, the tractability of non-cooperative games can be an issue, leading researchers to, for example, ``issue-by-issue'' analysis approaches of multi-attribute games \citep{lai2004literature}.

Here we extend assurance games to a general achievement game that allows for several goals and multiple strategy options per goal for each each agent. No requirement is made, that the utility of each agent with respect to each goal is identical (in other words symmetry is not required), nor is required that the agents agree on a common strategy (in other words a non-cooperative game-theoretic model is used). This is, to the best of the author's knowledge, a first published description of a multi-goal assurance game.

The \textit{importance of being different} theorem is presented, that states that a group of $N$ agents faced with achieving $N$ goals, and individual motivations that are such that each agent is uniquely motivated to spend the effort to solve one particular goal leads to one unique Nash equilibrium point that corresponds to a situation in which all goals are achieved. These theoretical results are augmented by computer results corresponding to group sizes of 2 to 5 faced with the challenge of achieving 1, 2 or 3 goals are presented. Both the theoretical and the computer results indicate that multiple goals and motivation asymmetry can facilitate the achievement of goals without the requirement for iteration or other mechanisms such as reciprocity. 

\section{The Multi-Goal Achievement Game}
Let there be a scenario in which a group of $N$ agents is faced with the challenge of achieving $M$ different goals. A formal definition of an $N$-agent, $M$-goal achievement game is presented below; note that the function $\Theta: \mathbb{R} \rightarrow \{0,1\} $ stands for a variant of the Heaviside step function: $\Theta(x) = 1$, if $x \geq 0$ and $\Theta(x) = 0$, if $x < 0$.

\begin{definition}{Multi-goal achievement game.}
Let there be a set of $N \in \mathbb{N}$ different agents and a set of $M \in \mathbb{N}$ different goals and a number $K \in \mathbb{N}, K > 1$. Each agent $i$ can choose between contributing an element of cost set $C = \{c_k | k \in \{1,\ldots,K\}, c_k \in \mathbb{R}, c_k \geq 0, c_{k_1} < c_{k_2} \iff k_1 < k_2\}$ towards any of the $M$ different goals. The chosen contribution of agent $i$ towards goal $j$ is denoted as $d_{ij} \in C$. Let $D$ be the $N \times M$ matrix consisting of the elements $d_{ij}$.  Let there be an $M$-tuple of positive goal thresholds $T=(g_1, g_2, \ldots, g_M) \in \mathbb{R}^M$. We say goal $j$ is achieved, if and only if $\sum\limits_{i=1}^N d_{ij} \geq g_j$. Let the utility of agent $i$ be the negative of the sum of effort units spent by agent $i$ plus a weighted sum of achieved goals, in other words $u_i(D) = - \sum\limits_{j=1}^{M}d_{ij} + \sum\limits_{j=1}^M w_{ij}\Theta(\sum\limits_{k=1}^{N} d_{kj} - g_j)$ with $w_{ij} \in \mathbb{R}$. Let $W$ be the $N \times M$ matrix consisting of the elements $w_{ij}$. We call the matrix $W$ the motivation matrix of the game. We call the finite, non-iterative $N$-player game $G(N, M, C, T, W)$ an $M$-\textup{goal achievement game} or an $N$-\textup{player}, $M$-\textup{goal}, $K$-\textup{choice achievement game}. If $M > 1$, we call the game a \textup{multi-goal achievement game}, otherwise a \textup{single-goal achievement game}. 
\end{definition}

Note that the reward of an agent $i$ to achieve a particular goal $j$ (represented by motivation matrix elements $w_{ij}$) can consist of a material reward or a subjective motivation or combinations thereof. The motivation matrix elements can nonetheless be measured in currency units (not in terms of financial rewards but in terms of the willingness to pay for achieving a certain goal).

\begin{lemma}
For an $N$-player, $M$-agent, $K$-choice achievement game, there are $K^{MN}$ different combinations of strategies.
\end{lemma}

\begin{proof}
Each player has $K$ choices for each of the $M$ different goals. Each player has thus overall $K^M$ different strategies to choose from. Because there are $N$ different players who can choose their strategy independently, there are  $(K^M)^N=K^{MN}$ different combinations of strategies.
\end{proof}

Agents may be motivated to contribute towards achieving certain goals, but it will not make sense for them to pay more than is needed to achieve the goals that are important to them:

\begin{lemma}
For any agent $i$ and any goal $j$, any strategy $a$ of spending $c_k$ effort units (with $c_k > g_j)$ on goal $j$ is strictly dominated by the strategy $a'$ of spending $c_{k'}$ effort units on goal $j$ (and unchanged effort units on all other goals $j' \neq j$) if $\exists k' \in \{1,\ldots,k-1\}$ such that $g_j \leq c_{k'} < c_k$.
\end{lemma}

\begin{proof}
Strategy $a'$ strictly dominates strategy $a$, because there is no difference in goal achievement for any goal (goal $j$ is achieved in either strategy) and strategy $a$ corresponds to spending $c_k-c_{k'} > 0$ more effort units compared to strategy $a'$. The difference in utility is $u(a')-u(a)=-c_{k'} - (-c_k) = c_k - c_{k'} > 0$. It follows that $u(a')>u(a)$, in other words the utility of strategy $a'$ is greater than the utility of strategy $a$.
\end{proof}

In this paper, a focus is on the interesting special case of multi-goal achievement games where the number of agents is equal to the number of goals with the additional provision that there is a one-to-one relationship between agents and the goals they are highly motivated to achieve. We call such games \textit{individual purpose games}:

\begin{definition}{Individual purpose game.}
Let there be an achievement game $G$ where the number of goals $M$ is equal to the number of agents $N$. Let $G$ be also such that the contribution $w_{ij}$ to the utility of agent $i$ for achieving goal $j$ is greater than the goal-achievement cost $g_j$ for $i = j$ and lower compared to $g_j$ for $i \neq j$. We call an achievement game with such properties an \textup{individual purpose game}. 
\end{definition}

We denote an individual purpose game as $G(N,C,T,W)$ where $C$ is a set of contribution choices that each agent can pay towards each goal, $T$ is an $M$-tuple of goal-thresholds and $W$ is an $N \times N$ matrix of goal-achievement rewards to agent's utility functions. This definition encompasses cases where the difficulty to achieve the different goals varies widely. We call the special case where the goal-achievement thresholds are equal between all goals an \textit{even individual purpose game}:

\begin{definition}{Even individual purpose game.}
Let there be an individual purpose game $G(N,C,T,W)$. If the goal threshold $N$-tuple $T$ is of the form $T=(g,g,\ldots,g)\in \mathbb{R}^N$, we call the game $G$ an \textup{even} individual purpose game. We call the value of $g$ the \textup{universal goal threshold} of the game.
\end{definition}

In some cases, the motivations of the agents towards their non-favorite goals are low; we define such occurrence as extreme individual purpose game:

\begin{definition}{Extreme individual purpose game.}
We say an individual purpose game is \textup{extreme} if and only if the motivation matrix elements $w_{ij}$ are such that $w_{ij}$ is lower compared to the cost difference between the second-lowest and the lowest-cost action for all agents $i$ and goals $j$ with $i \neq j$.
\end{definition}

Intuitively, we may suspect that agents in an individual purpose game tend to contribute more towards the goals for which they have a high motivation. Indeed, we are able to show this formally for a special kind of extreme even individual purpose game in form of the importance of being different theorem. It turns out that the only equilibrium solution is such that agents contribute substantially towards the one goal that is most important to them and nothing to the goals that are not important to them. While this strategy is highly asymmetric (a case of anti-coordination), it has the property that all goals are achieved:

\begin{theorem}{Importance of being different theorem.}
Let there be an extreme even individual purpose game $G$ where the lowest contribution $c_1$ to a goal is zero and the highest contribution $c_K$ is equal to the universal goal threshold $g$ of the game. The game $G$ has one and only one Nash equilibrium in which each agent $i$ chooses to spend $g$ effort units towards goal $i$ and zero effort units towards all other goals ($d_{ij} = g$ if $i=j$ and $d_{ij}=0$ if $i \neq j$). This equilibrium is a pure-strategy equilibrium in which all goals are achieved.
\end{theorem}

\begin{proof}

We will show that the described choice of strategies is the only Nash equilibrium of the game by showing that for each agent, this strategy strictly dominates all alternative strategies. We will do that by iteratively eliminating dominated strategies.

For a strategy, in which agent $i$ spends an effort greater than $c_1$ on a goal $j$, $i \neq j$, that agent can, independent of the choices of the other agents, increase its utility by instead spending $c_1$ effort units on goals $j$ because it decreases the spent effort units by at least $c_2-c_1$, does not change the accomplishments of goal $i$, and decreases the utility by a value less than $c_2-c_1$ due to the potential non-achievement of goal $j$ (because $w_{ij} < c_2-c_1$ for $i \neq j$). In other words, any strategy for which $d_{ij} > c_1$ for $i \neq j$ is strictly dominated by strategies for which $d_{ij} = c_1 = 0$ for $i \neq j$ and does not need to be considered further.

A strategy $a$ for which $0 \leq d_{ii} < g$, the utility of agent $i$ is $u_i(a) = -d_{ii} + x + y$, with $x = w_{ii}\Theta(\sum\limits_{k=1}^{N} d_{ki} - g)$ and $y = \sum\limits_{j=1, j\neq i}^M w_{ij}\Theta(\sum\limits_{k=1}^{N} d_{kj} - g)$. Because $d_{ij}=0$ for $i \neq j$, one obtains $\sum\limits_{k=1}^{N} d_{ki}=d_{ii}$ and $\sum\limits_{k=1}^{N} d_{kj} = d_{jj}$. It follows that $x = w_{ii}\Theta(d_{ii} - g)$ and $y = \sum\limits_{j=1, j\neq i}^M w_{ij}\Theta(d_{jj} - g)$. Because $d_{ii} < g$, it follows that goal $i$ is not achieved and $x = 0$, thus $u_i(a) = -d_{ii} + y \leq y$. If the effort $d_{ii}$ spent by agent $i$ on goal $i$ is, on the other hand, equal to $d_{ii}=g$, the utility for agent $i$ is $u_i(a') = -d_{ii} + w_{ii}\Theta(d_{ii} - g) + y = -g + w_{ii} + y$. Because $w_{ii} > g$ it follows that $u_i(a') > y \geq u_i(a)$.  In other words, the choice of agent $i $ spending $g$ effort units on goal $i$ strictly dominates the choice of spending less than $g$ effort on goal $i$. This leaves for each agent $i$ exactly one strictly dominating strategy of spending $g$ effort units on goal $i$ and zero effort units on all goals other than goal $i$. From this it follows that the strategy profile of each agent $i$ spending $g$ effort units on goal $i$ and zero effort units on all goals other than goal $i$ corresponds to the one and only one Nash equilibrium of the game. The identified Nash equilibrium is a pure-strategy equilibrium. Because $g$ effort units are spent on each goal $i$, all goals are achieved.
\end{proof}

A scenario that illustrates this situation in which a group of $N$ children is asked to feed and groom $N$ pet animals. If one child chooses to not put in the required effort, the other children might be tempted to also stop spending the effort to contribute to the upkeep of the group of pet animals. If, on the other hand, a child is ``in love'' with a pet animal, it will keep up the maintenance of that animal no matter how the other children are acting. If each child is ``in love'' with one unique animal, all $N$ animals are being cared for such that each child takes care of the pet animal it is ``in love'' with, while not contributing to the maintenance of the other pet animals.

\section{The Computational Approach}

Variants of goal achievement games are used for the computational results, in which each agent has for each goal the choice of contributing $0, \frac{1}{M}$ or $1$ effort units (payments).  These three choices are called defection, cooperation and heroic effort respectively. A specific goal is achieved, if the total payments towards that goal is at least $N/M$. To achieve all goals, the sum of all payments of the agents has thus to be at least $N$ units. The goal-specific utility of agent $i$ with respect to goal $j$ is the negative of the payment that the agent has performed towards goal $j$, plus a ``motivation'' term $w_{ij}$ that is added to the utility of an agent provided that goal $j$ is achieved. The utility of an agent is the sum of its goal-specific utilities. The motivations of the $N$ agents towards achieving the $M$ goals is defined through the $N \times M$ motivation matrix $W$. The matrix elements are for our numerical analysis set equal to one of the costs of the three possible choices plus a small excess motivation term $\delta$ (set to 0.25 utility units). These scenarios correspond to generalized achievement games and are analyzed using a game-theoretic approach.

Four different performance scores are defined that measure who well a group is able to achieve goals. These four scores are:
\begin{itemize}
\item The mean-goal-achievement score (MGA) is the mean of achieved goals, averaged over the different Nash equilibria of the achievement game.
\item The all-goal-achievement score (ALL) score for one Nash equilibrium of a achievement game is equal to 1 if all goals are achieved and zero otherwise. The ALL score of a game is the average over the ALL scores of all equilibrium points.
\item The defection-robustness score (DD) is the average of the MGA scores over all possible scenarios, in which exactly one agent is replaced by an agent with low motivation (set to $\delta$) towards all goals. This score measures the robustness of goal achievement under the challenge of additional defectors.
\item Variable-load score (VL) is the MGA score of a achievement game average over all scenarios in which the threshold to achieve one goal is instead of $n$ set to either $n+1$ or $n-1$ (with $n$ being equal to the number of agents). This score measures the robustness in goal achievement with respect to the challenge of a variable difficulty in achieving each goal.
\end{itemize}

We define a measure of group ``polarization'': The \textit{divergence} of two agents is the squared Euclidian norm of the difference of their motivation $M$-tuples divided by the number of goals. The divergence $V$ of a group is defined as the maximum divergence of any pair of its agents: 

$V = \frac{1}{M}\max\limits_{i,j \in \{1,\ldots,N\}} \sum_{k=1}^{M}{(w_{ik}-w_{jk})^2}$

Computational results were generated examining numerically identified pure-strategy Nash-equilibrium points of achievement games and iterating over different types of groups and number of goals. The analysis of an $N$-player $M$-goal achievement game has been implemented in the Java programming language. The computer program performs the creation of different group scenarios, the computation of utility matrices, the identification of pure-strategy Nash equilibiria and the computation of the different scores. Mixed-strategy equilibria are not considered. Group sizes of 2 to 5 members faced with the challenge of achieving one, two or three goals have been analyzed. The case of a group with five members faced with the achievement of three goals has not been analyzed numerically, because of its large computational cost of iterating over all possible group member motivations.

\section{Results}

\begin{figure}
\includegraphics[width=129mm]{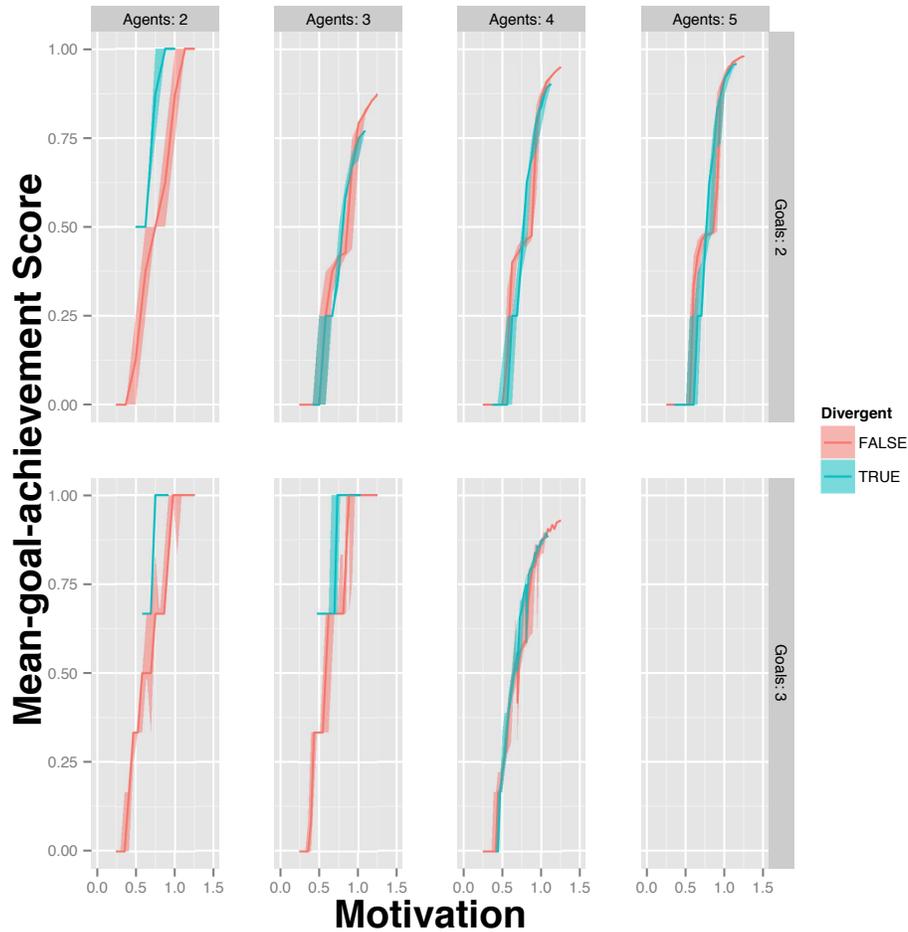}
\caption{Shown is the mean-goal-achievement (MGA) score as a function of mean motivation for groups consisting of two to five members faced with the achievement of two or three goals. The groups have been split in ``divergent'' and ``non-divergent'' groups, corresponding to whether they have a high or low maximum difference in goal priorities between any two group members. The mean-goal-achievement score is the fraction of goals that are achieved averaged over the pure-strategy Nash equilibria. A corresponding scatter plot that includes the case of one goal is shown in Figure 3.}
\end{figure}

The motivation of an agent with respect to one goal is defined here as being equal to the maximum effort that the agent would be willing to spend in order to reach that goal. Motivation is thus a measurable ``currency'' for voluntary efforts. One can define for each agent the sum of all motivations  towards all goals. The sum the total motivations leads to the total motivation of a group, a measure of it's capability to achieve goals. Dividing that number by the number of group members leads to a quantity that makes groups of different sizes comparable (called here mean motivation). Figure 1 depicts the mean-goal-achievement score as a function of a group's mean motivation for different number of agents and number of goals. As expected, the mean-goal-achievement score tends to be higher for groups with higher mean motivation. The groups are considered ``divergent'' or ``not-divergent'', depending on the size their maximum difference in goal priorities. One can see that for cases in which the number of agents is equal to the number of goals (2 and 3), the divergent groups tend to outperform the non-divergent groups.

Differences in goal achievement between divergent and non-divergent groups are depicted in Figure 2. Figures 1 and 2 demonstrate that for the case of 2 goals and more than 2 agents, one can identify regions, where divergent groups tend to do better, worse, or similar to non-divergent groups. 

\begin{figure}
\includegraphics[width=129mm]{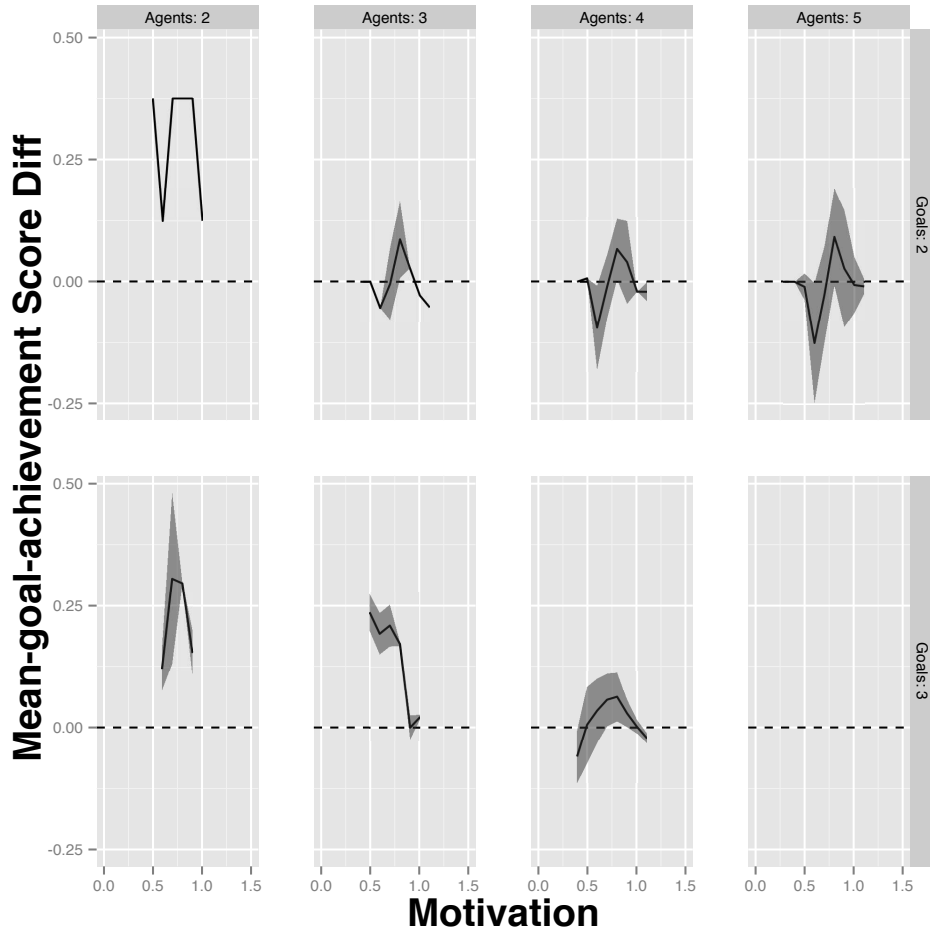}
\caption{Top-performing divergent groups frequently outperform top-performing non-divergent groups, especially if the number of goals is equal to the number of group members. Shown is the difference in mean-goal-achievement (MGA) score of top-performing divergent and non-divergent groups as a function of total group motivation. A positive value indicates, that top-performing divergent groups score higher compared to top-performing non-divergent groups. For the analysis, intervals with a width of 0.1 with respect to a group's mean motivation have been chosen. For each interval, the difference between the MGA score of the top-performing divergent group and the top-performing non-divergent group has been computed. The width of the ribbon is computed as the sum of the median absolute deviations (mad) values of the MGA scores for the compared sets of divergent and non-divergent groups.}
\end{figure}

In addition to the question of what fraction of goals a particular group can be expected to achieve, one can ask the question of how likely it is, that a group reaches all goals. Alternatively one can analyze, how well the goal achievement will be under the challenge that any one group member is replaced by someone who is unmotivated with respect to all goals. The third alternative score is the average goal achievement score under the additional challenge of variable goal difficulty. Scatter plots for the MGA-score and for the three alternative performance scores (termed ALL, DD and VL) are shown in Figures 3-6. One can see that overall tendencies for the alternative scores are not dramatically different compared to the MGA scores depicted in Figures 1.

\begin{figure}

\begin{tabular}{cc}
\includegraphics[width=70mm]{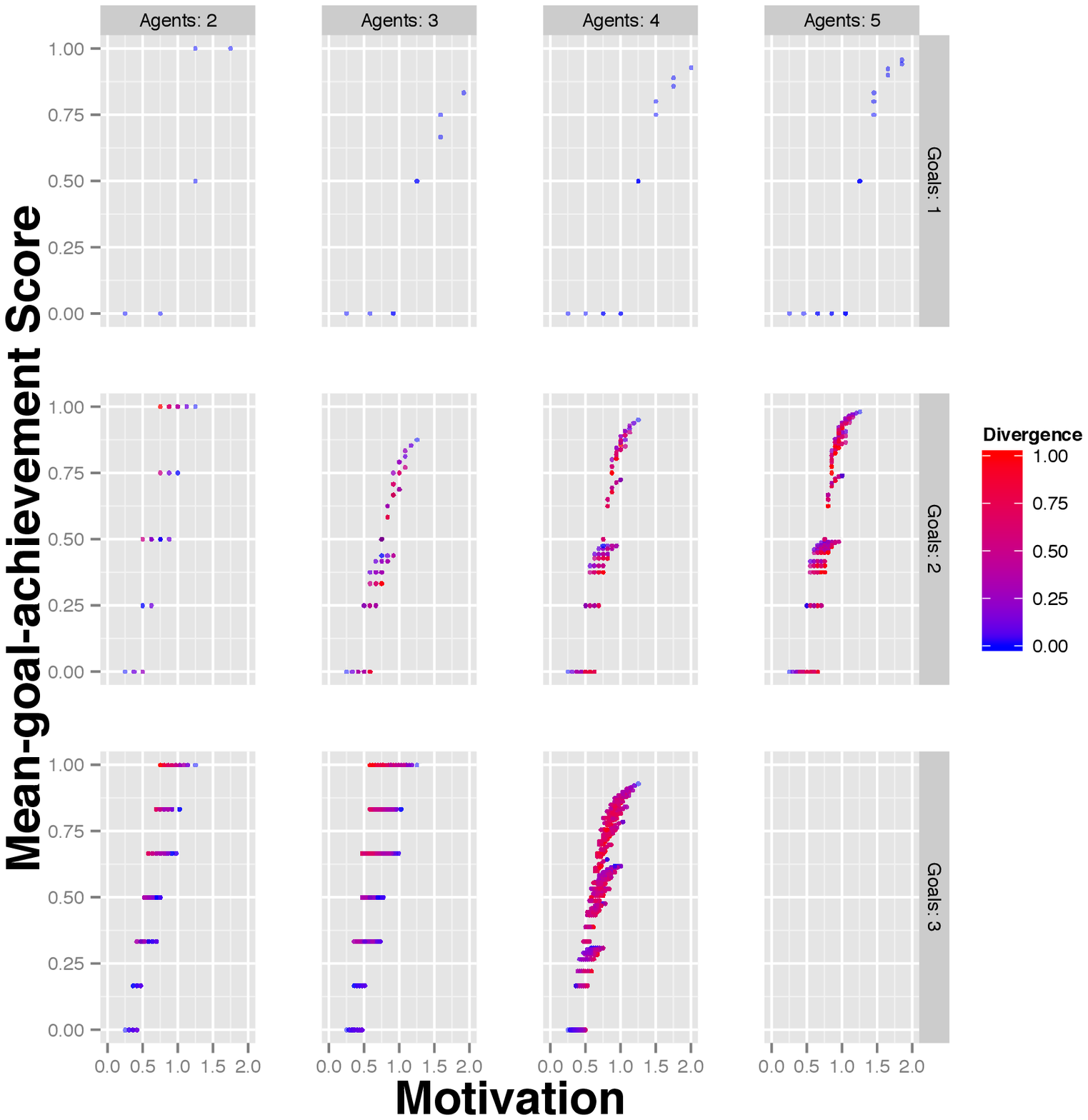} & \includegraphics[width=70mm]{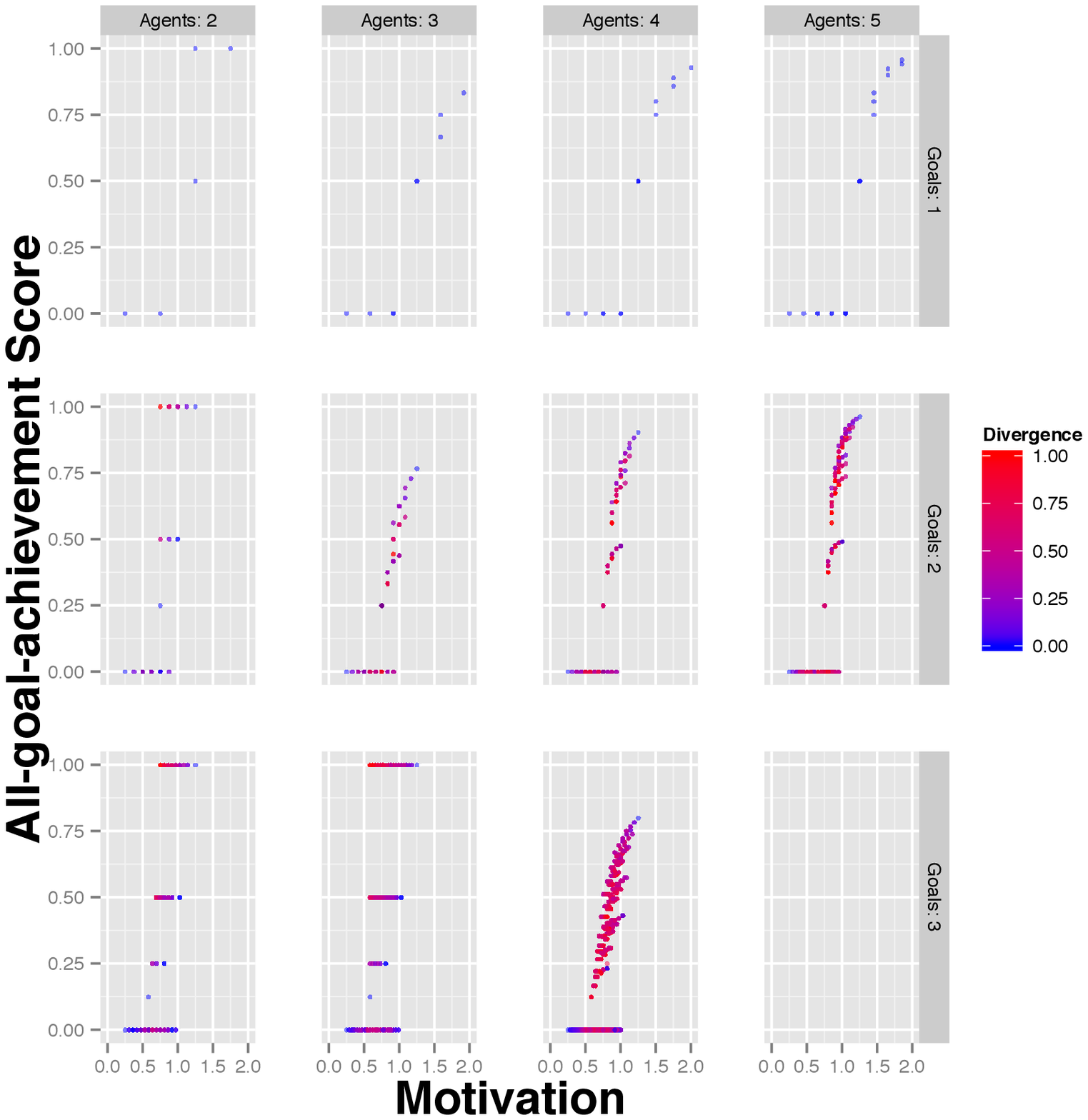} \\

\includegraphics[width=70mm]{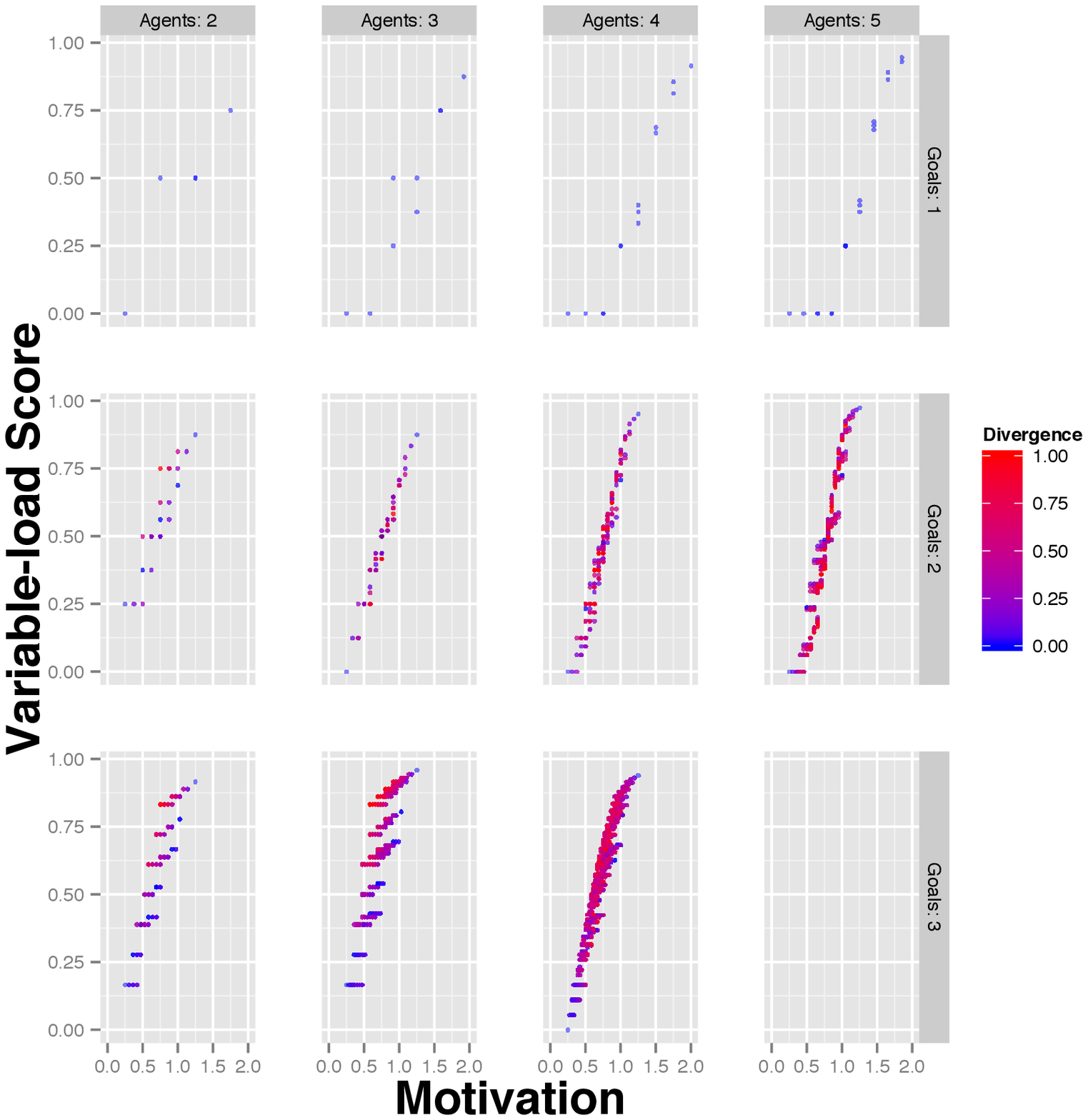} & \includegraphics[width=70mm]{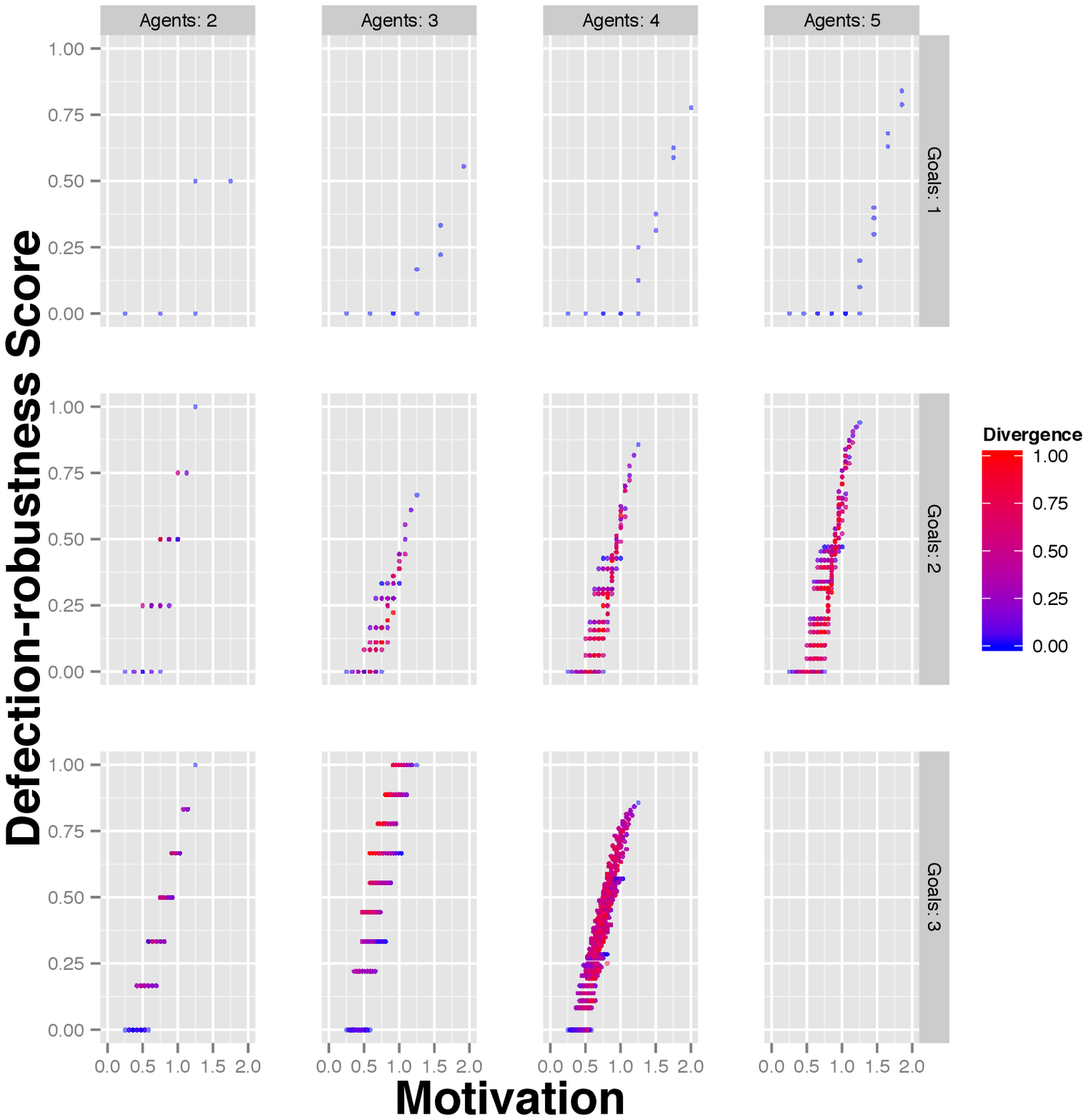} \\

\end{tabular}

\caption{ Shown are different goal achievment scores as a function of mean motivation for groups consisting of two to five members faced with the achievement of one, two or three goals. Top left: mean-goal-achievement score (the fraction of goals that are achieved averaged over the pure-strategy Nash equilibria). Top right: all-goal score (the all-goal score is the average is the fraction of Nash equilibria in which all goals are achieved). Bottom-left: variable-load score (the average in mean goal achievement for all scenarios for which one goal has one unit higher cost of achievement and all scenarios in which one goal has one unit lower cost of achievement). Bottom-right: defection-robustness score (the average in mean goal achievement for all scenarios for which one group member is replaced by an agent who has low motivation to achieve any goals). }
\end{figure}

In Tables 1 and 2, the different score of groups consisting of two and four members respectively that are faced with two goals are listed. Each agent can be one of three types: $A$ stands for a high motivation towards goal one, and a low motivation towards goal 2. $B$ stands for agents who have a high motivation with respect to goal 2 and a low motivation with respect to goal 1. Type $O$ stands for agents, who have a medium motivation with respect to both goals. Shown are the mean-goal-achievement score (MGA), the all-goal-achievement score (ALL), the defection robustness score (DD) and the variable-load-score (VL) as well as the corresponding ranks (indicated as MGAR, ALLR, DDR and VLR respectively). The low DD score of 0.0 for the OOOO group (a group consisting of four ``centrist'' members that have medium motivation towards achieving each of the two goals) indicates that a `centrist` group is doing relatively well under regular conditions (an MGA-score value of 0.5) but is highly vulnerable to the introduction of an agent with low motivation that leads to a break-down of goal achievement. The semi-polarized (AOOB) and polarized (AABB) groups perform well under unchallenged conditions (MGA score rank 1 tied with the OOOO group) but perform also well under conditions in which the achievement of goals is under variable bias (the two groups have the highest-ranking variable load score). Surprisingly well-performing are the highly biased groups AAAA and BBBB. Those groups consists of members that are highly motivated to achieve only the first or only the second goal, respectively. These groups essentially focus on solving only one of the two goals while ignoring the other goal. These highly biased groups are less vulnerable to the introduction of a ``black sheep'' that has low motivation to achieve either goal.

\begin{table}[ht]
\centering
\begin{tabular}{lrrrrrrrrrr}
  \hline
Motivations & MGA & ALL & DD & VL & MGAR & ALLR & DDR & VLR & Wins & Ties \\
  \hline
AB & 1.00 & 1.00 & 0.50 & 0.75 &   1 &   1 &   1 &   1 & 5.00 & 0.00 \\
  BB & 0.50 & 0.00 & 0.50 & 0.56 &   2 &   3 &   1 &   2 & 3.00 & 1.00 \\
  AA & 0.50 & 0.00 & 0.50 & 0.56 &   2 &   3 &   1 &   2 & 3.00 & 1.00 \\
  AO & 0.50 & 0.00 & 0.25 & 0.50 &   2 &   3 &   4 &   4 & 0.00 & 2.00 \\
  OB & 0.50 & 0.00 & 0.25 & 0.50 &   2 &   3 &   4 &   4 & 0.00 & 2.00 \\
  OO & 0.50 & 0.25 & 0.00 & 0.50 &   2 &   2 &   6 &   4 & 0.00 & 2.00 \\
   \hline
\end{tabular}
\caption{2-goal, 2-member achievement game, in which each member can be of type A,B or O, corresponding to motivations towards goals 1 and 2 being $(1+\delta, \delta), (\delta, 1+\delta), (0.5+\delta, 0.5+\delta)$ respectively (with $\delta=0.25$). MGA: the score assigned to a group is the fraction of achieved goals averaged over the equilibrium points. ALL: number of equilibria in which all goals are achieved divided by the total number of equilibrium points; DD: the MGA score, averaged over groups in which the motivation of one group member is set to $(\delta, \delta)$. VL: the average of the MGA score for which the threshold to achieve one goal is one unit higher or one unit lower; MGAR,ALLR,DDR,VLR: the rank of group among the listed groups with respect to the MGA,ALL,DD or VL score, respectively. Tied ranks are replaced with the minimum tied rank. Wins: number of cases, in which a group has better ranks with respect to more scores compared to another group. Ties: number of cases, in which a group has tied score ranks compared to another group. One can see that group AB is outperforming (or tied with) the other groups with respect to all scores. In other words, groups of two, in which one member is highly motivated with respect to one goal, while the other group member is highly motivated with respect to the other goal are more likely to achieve both goals (ALL score of 1.0) compared to the other listed groups.}
\end{table}

\begin{table}[ht]
\centering
\begin{tabular}{lrrrrrrrrrr}
  \hline
Motivations & MGA & ALL & DD & VL & MGAR & ALLR & DDR & VLR & Wins & Ties \\
  \hline
AABB & 0.50 & 0.25 & 0.25 & 0.50 &   1 &   1 &   9 &   3 & 12.00 & 2.00 \\
  AOOB & 0.50 & 0.25 & 0.12 & 0.50 &   1 &   1 &  14 &   3 & 11.00 & 2.00 \\
  AAAA & 0.47 & 0.00 & 0.43 & 0.48 &   4 &   4 &   1 &   5 & 10.00 & 2.00 \\
  BBBB & 0.47 & 0.00 & 0.43 & 0.48 &   4 &   4 &   1 &   5 & 10.00 & 2.00 \\
  OBBB & 0.46 & 0.00 & 0.39 & 0.46 &   6 &   4 &   3 &  10 & 8.00 & 1.00 \\
  AAAO & 0.46 & 0.00 & 0.39 & 0.46 &   6 &   4 &   3 &  10 & 8.00 & 1.00 \\
  OOOO & 0.50 & 0.25 & 0.00 & 0.46 &   1 &   1 &  15 &   9 & 6.00 & 6.00 \\
  AAOO & 0.44 & 0.00 & 0.31 & 0.47 &   8 &   4 &   5 &   7 & 6.00 & 2.00 \\
  OOBB & 0.44 & 0.00 & 0.31 & 0.47 &   8 &   4 &   5 &   7 & 6.00 & 2.00 \\
  ABBB & 0.43 & 0.00 & 0.29 & 0.53 &  10 &   4 &   7 &   1 & 4.00 & 4.00 \\
  AAAB & 0.43 & 0.00 & 0.29 & 0.53 &  10 &   4 &   7 &   1 & 4.00 & 4.00 \\
  OOOB & 0.40 & 0.00 & 0.19 & 0.42 &  12 &   4 &  10 &  14 & 2.00 & 1.00 \\
  AOOO & 0.40 & 0.00 & 0.19 & 0.42 &  12 &   4 &  10 &  14 & 2.00 & 1.00 \\
  AOBB & 0.38 & 0.00 & 0.16 & 0.44 &  14 &   4 &  12 &  12 & 0.00 & 1.00 \\
  AAOB & 0.38 & 0.00 & 0.16 & 0.44 &  14 &   4 &  12 &  12 & 0.00 & 1.00 \\
   \hline
\end{tabular}
\caption{Performance scores of groups with 4 members that are faced with the challenge of achieving two goals (called goal 1 and 2). Each group member can be of type A,B or O corresponding to a high (low), low (high) or medium (medium) motivation with respect to goal 1 (2). The score definitions are described in the caption of Figure 1. One can see, that the (semi)-``polarized'' groups AABB and AOOB score well in comparison to the ``balanced'' group OOOO. Surprisingly well perform groups AAAA and BBBB that essentially achieve one goal reliably while ignoring the other goal.} 
\end{table}

\begin{table}[ht]
\centering
\begin{tabular}{lrrrrrrrrrr}
  \hline
Motivations & MGA & ALL & DD & VL & MGAR & ALLR & DDR & VLR & Wins & Ties \\
  \hline
AOB & 0.50 & 0.25 & 0.17 & 0.50 &   1 &   1 &   5 &   3 & 7.00 & 2.00 \\
  OOO & 0.50 & 0.25 & 0.00 & 0.50 &   1 &   1 &  10 &   3 & 6.00 & 2.00 \\
  AAA & 0.44 & 0.00 & 0.33 & 0.44 &   3 &   3 &   1 &   5 & 6.00 & 1.00 \\
  BBB & 0.44 & 0.00 & 0.33 & 0.44 &   3 &   3 &   1 &   5 & 6.00 & 1.00 \\
  OBB & 0.42 & 0.00 & 0.28 & 0.52 &   5 &   3 &   3 &   1 & 4.00 & 3.00 \\
  AAO & 0.42 & 0.00 & 0.28 & 0.52 &   5 &   3 &   3 &   1 & 4.00 & 3.00 \\
  AOO & 0.38 & 0.00 & 0.17 & 0.44 &   7 &   3 &   5 &   5 & 2.00 & 1.00 \\
  OOB & 0.38 & 0.00 & 0.17 & 0.44 &   7 &   3 &   5 &   5 & 2.00 & 1.00 \\
  ABB & 0.33 & 0.00 & 0.11 & 0.42 &   9 &   3 &   8 &   9 & 0.00 & 1.00 \\
  AAB & 0.33 & 0.00 & 0.11 & 0.42 &   9 &   3 &   8 &   9 & 0.00 & 1.00 \\
   \hline
\end{tabular}
\caption{Performance of groups of 3 members faced with a challenge of achieving two goals and having equal total motivations.}
\end{table}

\section{Discussion}

The obtained results demonstrate that the ranking of a group with respect to goal achievement depends on the utilized scoring method as well as on the influence of additional challenges (an additional-defection challenge and a variable-load challenge were examined). The results suggests, that (partial) polarization of a group can have advantages for robust, compartmentalized problem-solving of a group. This is particularly interesting for the achievement of goals, for which only a minority is motivated to spend efforts for its solution. In such cases, previously reported approaches for reaching cooperation (such as coercion, kin-selection, costly punishment, reputation, reciprocity etc.) do not easily apply \citep{Fehr1999,Fehr2000,Dohmen2006,Panchanathan2004,Fowler2005}. Indeed, such mechanisms favoring behavioral similarity may work \textit{against} someone who is attempting to exert individual efforts towards achieving a non-profit goal that is considered important only by a minority.

These results were obtained under the simplifying assumption, that someone who is highly biased towards one goal does not attempt to hinder attempts by others to achieve a different goal. Also, one should be guarded to extrapolate the obtained results for groups with up to five members to groups of dramatically larger sizes. Furthermore, we know through reported experimental results that groups do not necessarily choose the ``best'' equilibrium point, nor do they choose equilibrium points with equal probability \citep{VanHuyck1990}.

The results make it, despite the conceptual shortcomings, tempting to speculate that the ubiquitous phenomenon of a left-wing/right-wing political spectrum can be understood as a special case in which the two commons to be maintained by a society are a social network (helping others) and a social shield (protecting from others). This situation is called here the Great Tale of Two Commons (GTTC). The theoretical results suggest that groups whose political spectrum is narrow (only moderates) are less engaged in voluntary efforts compared to semi-polarized groups that partially consist of participants who are biased towards one or the other goal.

This is a different paradigm compared to the model of Dixit and Weibull, in which political polarization is a result of a learning process and a consequence of failed policies \citep{Dixit2007}. The GTTC and the Dixit and Weibull model may be falsifiable: Years of successful policies in the light of relatively static demands on a society would lead in the GTTC model to maintained polarization while the Dixit and Weibull model predicts convergence. 

This strategy of compartmentalizing group-goals can be extended by adding further goals and respective member-specific motivations thus leading to higher-dimensional achievement games, including the extreme case of an achievement game in which the number of goals is equal to the number of group members. The one and only equilibrium reached in the case in which each group member is uniquely motivated to reach one particular goal is a fascinating case of goal achievement: because the chosen actions of the group members are different, it is an example of ``anti-coordination''. As can be seen by the results for groups of three members faced with the challenge of achieving two goals (Table 3), the three strategies of coordination, anti-coordination and prioritization (the ignorance of goals that are perceived to be less important) are all performing quite well. Top-performing larger groups can perform combinations of these strategies, as exemplified by the ``AABB'' group of four members who succeed in achieving two goals by two members focusing on the first goal and the two other members focusing on the second goal (see Table 2). This scenario has similarity to congestion or crowding games, in which agents attempt to avoid making similar choices, that could lead to the over-use of any one strategy \citep{rosenthal1973class}. In other words, one strategy to cope with the difficulty of coordination games may be to ``change the game'' such that the need for coordination is minimized.

Another example where a related strategy of problem-solving via high-dimensional anti-coordination appears to be used is parenting, whereby parents commit all necessary resources towards the upbringing of their own progeny, will contributing comparatively little to the upbringing of those who are not part of their family.

Applying this approach to the environmental situation, one needs to remind oneself of survey results that consistently show that environmental goals are viewed by the majority as less important as, for example, economic considerations \citep{upham2009public}. If goal 1 stands for a goal of ``economic growth'' and goal 2 stands for a goal of ``environmental protection'', the representative 4-member groups from Table 2 would thus be called ``AAAA'', ``AAAO'' or ``AAAB'', in other words groups consisting mostly of members who would give economic growth priority over environmental protection. Such groups perform relatively well by succeeding reliably in achieving the prioritized goal while ignoring the goal viewed as less important (the ALL score of Table 2 shows that the fraction of equilibrium points in which both goals are achieved is in all three cases zero). The opportunity suggested by the research presented in this paper is now to use a less ambitious second goal that can be achieved by the minority (``B-agents''), for which the second goal has high priority. Just like voluntary firefighters solve a minority problem (firefighting) by taking on their own shoulders the burden of solving the problem they feel strongly about even though they did not cause it. This mode of environmental protection would suggest a new type of environmental protectionists, who solve environmental problems by personally taking on the burden of solving it. Such an approach would be fundamentally different and yet complementary to current ``mainstream'' environmental approaches that attempt to convince all participants to give economic and environmental considerations similar priority (which applied to Table 2 would be, if successful, leading to ``OOOO'' groups). Such novel approaches could be combined with other recent innovations, such as the augmentation to the ``reduce, reuse, recycle'' paradigm in order to reduce rebound effects \citep{Restore}.

The mechanism proposed here corresponds to an incentive for altruistic behavior in terms of member's subjective utility function, not in terms of their objective burden. The results suggest that one of the strategies for achieving goals in a population may be to have a ``noisiness'' of motivations of its members towards multiple objectives, such that for each challenge a population faces, a subset of highly motivated members voluntarily emerges that personally takes on the burden of solving the challenge at hand. The results suggest, that differences in motivations and priorities with respect to life's various challenges we might have with someone, are indeed an opportunity and the result of Nature's approximate solution to escape mutual defection equilibria.

\bibliographystyle{unsrt}
\bibliography{heroic}		

\end{document}